\documentclass[12pt,twoside]{article}

\usepackage{float}
\usepackage{bm}
\usepackage{graphicx}
\usepackage{epstopdf}
\usepackage[figuresright]{rotating}
\usepackage{graphicx}
\usepackage{epic}
\usepackage{multirow}
\usepackage{tikz}
\usepackage{xcolor}
\usepackage{threeparttable}
\usepackage{ amssymb }
\usetikzlibrary{arrows,shapes,chains}

\usepackage{graphicx}
\usepackage{makecell}
\renewcommand{\paragraph}{\roman{paragraph}}
\usepackage[a4paper]{geometry}
\makeatletter
\renewcommand\title[1]{\gdef\@title{\reset@font\Large\bfseries #1}}
\renewcommand\section{\@startsection {section}{1}{\z@}%
                                   {-3.5ex \@plus -1ex \@minus -.2ex}%
                                   {2.3ex \@plus.2ex}%
                                   {\normalfont\large\bfseries}}
\renewcommand\subsection{\@startsection{subsection}{2}{\z@}%
                                     {-3ex\@plus -1ex \@minus -.2ex}%
                                     {1.5ex \@plus .2ex}%
                                     {\normalfont\normalsize\bfseries}}
\renewcommand\subsubsection{\@startsection{subsubsection}{3}{\z@}%
                                     {-2.5ex\@plus -1ex \@minus -.2ex}%
                                     {1.5ex \@plus .2ex}%
                                     {\normalfont\normalsize\bfseries}}

\def\@runningauthor{}\newcommand{\runningauthor}[1]{\def\runningauthor{#1}}
\def\@runningtitle{}\newcommand{\runningtitle}[1]{\def\runningtitle{#1}}

\renewcommand{\ps@plain}{%
\renewcommand{\@evenhead}{\footnotesize\scshape \hfill\runningauthor\hfill}
\renewcommand{\@oddhead}{\footnotesize\scshape \hfill\runningtitle\hfill}}

\newcommand{\F}{\mathbb{F}}

\newcommand{\lcm}{\rm lcm}

\newcommand {\C}{{\mathcal{C}}}

\newcommand {\supp}{{\mathrm{supp}}}
\newcommand {\Res}{{\mathrm{Res}}}
\g@addto@macro\bfseries{\boldmath}

\makeatother

\usepackage{amsthm,amsmath,amssymb}
\usepackage{cite}
\usepackage{graphicx}

\usepackage[colorlinks=true,citecolor=black,linkcolor=black,urlcolor=blue]{hyperref}

\theoremstyle{plain}
\newtheorem{theorem}{Theorem}[section]

\newtheorem{lem}[theorem]{Lemma}
\newtheorem{cor}[theorem]{Corollary}

\theoremstyle{definition}
\newtheorem{definition}[theorem]{Definition}
\newtheorem{example}[theorem]{Example}

\newtheorem{problem}{Problem}

\theoremstyle{remark}
\newtheorem{remark}[theorem]{Remark}

\runningauthor{}

\date{}

\begin{document}
\begin{sloppypar}

\title{On Galois self-orthogonal algebraic geometry codes
\thanks{This research is supported by the National Natural Science Foundation of China under Grant No. 12171134 and U21A20428. ({\em Corresponding author: Shixin Zhu})} 
\author{Yun Ding, Shixin Zhu, Xiaoshan Kai, Yang Li}
\thanks{Yun Ding, Shixin Zhu, Xiaoshan Kai and Yang Li are with School of Mathematics, 
Hefei University of Technology, Hefei, China 
(email: yundingmath@163.com, zhushixinmath@hfut.edu.cn, kxs6@sina.com, yanglimath@163.com).}}

\maketitle

\begin{abstract}
    Galois self-orthogonal (SO) codes are generalizations of Euclidean and Hermitian SO codes.
    Algebraic geometry (AG) codes are the first known class of linear codes exceeding the Gilbert-Varshamov bound. 
    Both of them have attracted much attention for their rich algebraic structures and wide applications in these years.  
    In this paper, we consider them together and study Galois SO AG codes. 
    A criterion for an AG code being Galois SO is presented. Based on this criterion, we construct 
    a new class of maximum distance separable (MDS) Galois SO AG codes from projective lines 
    and several new classes of Galois SO AG codes from projective elliptic curves, hyper-elliptic curves and Hermitian curves.  
    In addition, we give an embedding method that allows us to obtain more MDS Galois SO codes from known MDS Galois SO AG codes.  
\end{abstract}

\section{Introduction}
Throughout this paper, let $\F_q$ be the finite field with size $q$, where $q=p^h$ and $p$ is a prime. 
Let $\F_q^*=\F_q\setminus \{0\}$. 
Denote the $n$-dimensional vector space over $\F_q$ by $\F_q^n$. 
Then a $q$-ary {\em linear code} $\C$ with length $n$, dimension $k$ and minimum distance $d$, 
denoted by $[n,k,d]_q$, is a $k$-dimensional linear subspace of $\F_q^n$ and it can correct 
up to $\lfloor \frac{d-1}{2} \rfloor$ errors. 
A linear $[n,k,d]_q$ code $\C$ is said to be {\em maximum distance sparable} (MDS) if $d=n-k+1$. 
Since MDS codes have maximum error-correcting capability for given lengths and dimensions, 
it is always an interesting problem to construct different classes of MDS codes. 

Let $e$ be an integer satisfying $0\leq e\leq h-1$. 
Let $\bm{x}=(x_1,x_2,\ldots,x_n)$ and $\bm{y}=(y_1,y_2,\ldots,y_n)$ be any two vectors in $\F_q^n$. 
Fan and Zhang \cite{Galois.1} introduced the {\rm $e$-Galois inner product} as follows:  
$$\langle  \bm{x},\bm{y} \rangle _e=\sum _{i=1}^{n}x_iy_i^{p^e},~{\rm where}~0 \leq e \leq h-1.$$
The {\em $e$-Galois dual code} of a linear $[n,k]_q$ code $\C$ is given by 
\begin{equation*}
    \C^{\perp _e}=\{\bm{y} \in \F_q^n \mid  \langle \bm{x},\bm{y} \rangle _e =0,\ \forall \  \bm{x} \in \C\} 
\end{equation*}
and $\C^{\perp_e}$ is a linear $[n,n-k]_q$ code. 
For a linear  $[n,k]_q$ code $\C$, it is called {\em $e$-Galois self-orthogonal} (SO) if $\C \subseteq \C^{\perp_e}$. 
Note that if $e=0$ (resp. $e=\frac{h}{2}$ and $h$ is even), the $0$-Galois (resp. $\frac{h}{2}$-Galois with even $h$) inner product 
coincides with the {\rm Euclidean (resp. Hermitian) inner product}. 
Hence, one can see that $e$-Galois SO codes are generalizations of Euclidean SO codes 
(i.e., $0$-Galois SO codes) and Hermitian SO codes (i.e., $\frac{h}{2}$-Galois SO codes with even $h$). 

On one hand, 
extensive researches have established strong correlations between SO codes and other mathematical topics 
such as group theory \cite{lattice }, lattice theory \cite{lattice }, modular forms \cite{forms }, 
combinatorial $t$-design theory \cite{t-design }, and (entanglement-assisted) quantum error-correcting codes \cite{quantum.2,quantum.4,quantum.5,liyang1}. 
These outstanding works have also further stimulated the study of Galois SO codes. 
Recently, the classical (extended duadic) constacyclic codes \cite{con.1,con.3,con.4,Galois.1}, skew-twisted codes \cite{con.2 }, (extended) generalized Reed-Solomon codes \cite{G.hso, jinlinfei.hso,caomeng,liyang.gso} 
and Goppa codes \cite{Goppa.2022,Goppa.2023} have been to used to construct Galois SO codes.  
For references, we list known MDS Galois SO codes in Table \ref{tab:1}. 

On the other hand, Goppa \cite{AGcodes} introduced the so-called algebraic geometry (AG) codes in 1981 
and this family of linear codes includes generalized Reed-Solomon codes and Goppa codes mentioned above as two special subclasses \cite{H.2009}. 
AG codes are very interesting since they have a good lower bound for the minimum distance 
and can exceed the well-known Gilbert-Varshamov bound \cite{GVb}.
In particular, AG codes from projective lines are MDS codes \cite{H.2009}. 
For more details on AG codes, please refer to \cite{Goppa.4, chenhao.1, chenhao.2, chenhao.3,Ls.2023 hull} and the references therein.

However, unlike generalized Reed-Solomon codes and Goppa codes, general AG codes have been studied in the literature only 
as candidates for constructing Euclidean and Hermitian SO codes. 
Specifically speaking, Stichtenoth $et~al.$ \cite{H.so} constructed Euclidean SO AG codes from a certain optimal tower. 
Jin $et~al.$ gave sufficient conditions that an AG code was equivalent to an Euclidean (resp. Hermitian) SO code \cite{jinlinfei.agc} and
constructed Euclidean SO AG codes from elliptic curves \cite{jinlinfei.nmds}.
Hernando $et~al.$ \cite{F.qso} showed that there is a large family of curves from which to obtain Euclidean SO AG codes and this large family 
includes Castle curves \cite{Castle curves}. Sok \cite{sok.2021 sd,sok.2021 hso} also presented some new constructions on Euclidean and 
Hermitian SO AG codes from algebraic curves of zero genus, one genus and higher genus.
From the discussion above, a natural problem rises as follows:  
\begin{problem}\label{prob.1}
    Are there general Galois SO AG codes?  If they exist, how should we construct them?
\end{problem}  
Motivated by this problem, we study and construct Galois SO AG codes in this paper. 
Our main contributions can be summarized as follows: 
\begin{itemize}
    \item [\rm (1)] We give a criterion for an AG code being a general Galois SO code in Lemma \ref{t1}. 
    This criterion provides a positive answer to the first part of Problem \ref{prob.1}. 
    
    \item [\rm (2)] Over projective lines,  
    we first give an embedding method in Lemma \ref{l1}, which allows us to construct more MDS Galois SO codes from known MDS Galois SO AG codes. 
    Then we present a new explicit construction of MDS Galois SO AG codes in Theorem \ref{t3}. 
    
    \item [\rm (3)] Over projective elliptic curves, hyper-elliptic curves and Hermitian curves, 
    we also get some new Galois SO AG codes with good parameters in Theorems \ref{t5}, \ref{t6} and \ref{t7}. 
    We list these new (MDS) Galois SO AG codes in Table \ref{tab:2} and hence, 
    the second part of Problem \ref{prob.1} has been fixed in some degrees. 
\end{itemize}

The paper is laid out as follows. 
In Section \ref{sec2}, we review some basic notations and concepts on algebraic function fields and AG codes. 
In Section \ref{sec 30}, we first propose a criterion to determine if an AG code is Galois SO. 
Then, we present the embedding method and show a concrete construction of MDS Galois SO AG codes over projective lines.
In addition, we also construct some new Galois SO AG codes over projective elliptic curves, hyper-elliptic curves, and Hermitian curves.
Finally, Section \ref{sec5} concludes this paper. 

\begin{table}
    \caption{Some known MDS $e$-Galois SO codes }
       \label{tab:1}       
       \begin{center}
        \resizebox{\linewidth}{!}{
           \begin{tabular}{c c c c c}
            
            \hline
             Class & Finite field $\F_q$ & Length $n$ & Dimension $k$& Reference \\\hline
             \hline



            1 & $q=p^h$ is odd, $2e \mid  h$ & $n=\frac{t(q-1)}{p^e-1}+1$, $1 \leq t \leq p^e-1$ &$1 \leq k \leq \lfloor \frac{p^e+n}{p^e+1}\rfloor  $ & \cite{caomeng}\\ \hline

            2 & $q=p^h$ is odd, $2e \mid  h$ & $\begin{array}{c} n=r_1\frac{q-1}{\gcd(q-1,x_2)}+1, 1 \leq r_1 \leq \frac{q-1}{\gcd(q-1,x_1)},\\
                (q-1)\mid  {\lcm}(x_1,x_2), 
                \frac{q-1}{p^e-1} \mid  x_1 \end{array}$ & $1 \leq k \leq \lfloor \frac{p^e+n}{p^e+1}\rfloor  $ & \cite{caomeng}\\\hline

            3 & $q=p^h$ is odd, $2e \mid  h$ & $\begin{array}{c} n=rm+1,1 \leq r \leq \frac{p^e-1}{m_1},\\
                m \mid  (q-1),m_1=\frac{m}{\gcd(m,y)},
                y=\frac{q-1}{p^e-1} \end{array}$ & $1 \leq k \leq \lfloor \frac{p^e+n}{p^e+1} \rfloor $ & \cite{caomeng}\\ \hline

            4 &$q=p^h$ is odd, $2e \mid  h$ & $\begin{array}{c} n=tp^{aw},a \mid  e, 1 \leq t \leq p^a,\\
                1 \leq w \leq \frac{h}{a}-1 \end{array} $ & $1 \leq k \leq \lfloor \frac{p^e+n-1}{p^e+1}\rfloor  $ & \cite{caomeng}\\ 
                
                \hline

            5 & $q=p^h$ is odd, $2(e-\mu ) \mid  h$ & $n=\frac{t(q-1)}{p^{e-\mu}-1}+1$, $1 \leq t \leq p^{e-\mu}-1$ &$1 \leq k \leq \lfloor \frac{p^{\mu}+n}{p^{\mu}+1}\rfloor  $ & \cite{caomeng1}\\ \hline

            6 & $q=p^h$ is odd, $2(e-\mu) \mid  h$ & $\begin{array}{c} n=r\frac{q-1}{\gcd(q-1,x_2)}+1, 1 \leq r \leq \frac{q-1}{\gcd(q-1,x_1)},\\
                    (q-1)\mid  {\lcm}(x_1,x_2), 
                    \frac{q-1}{p^{e-\mu}-1} \mid  x_1 \end{array}$ & $1 \leq k \leq \lfloor \frac{p^{\mu}+n}{p^{\mu}+1}\rfloor  $ & \cite{caomeng1}\\\hline
    
            7 & $q=p^h$ is odd, $2(e-\mu)\mid  h$ & $\begin{array}{c} n=rm+1,1 \leq r \leq \frac{p^{e-\mu}-1}{m_1},\\
                    m \mid  (q-1),m_1=\frac{m}{\gcd(m,y)},
                    y=\frac{q-1}{p^{e-\mu}-1} \end{array}$ & $1 \leq k \leq \lfloor \frac{p^{\mu}+n}{p^{\mu}+1} \rfloor $ & \cite{caomeng1}\\ \hline
    
           8 &$q=p^h$ is odd, $2(e-\mu) \mid  h$ & $\begin{array}{c} n=tp^{aw},a \mid  {(e-\mu)}, 1 \leq t \leq p^a,\\
                    1 \leq w \leq \frac{h}{a}-1 \end{array} $ & $1 \leq k \leq \lfloor \frac{p^{\mu}+n-1}{p^{\mu}+1}\rfloor  $ & \cite{caomeng1}\\ 
                    
                    \hline

            9 &    $q=p^h$ & $ \begin{array}{c} n \mid  \gcd(p^{h-e}-1,\frac{q-1}{r}), \\
                    m_0=\lfloor \frac{n-1-s_0}{2} \rfloor, s_0=\frac{p^{h-e}+1}{r},\\
                    h_0=\begin{cases} -\lfloor \frac{s_0-1}{2}\rfloor  & \mbox{if $n$ is even},\\
                        -\lfloor \frac{s_0-2}{2} \rfloor & \mbox{if $n$ is odd} \end{cases} \end{array}$ & $ \begin{array}{c} k=m-h+1,\\
                           h_0 \leq h \leq m \leq m_0 \end{array}$ & \cite{con.3} \\

                \hline

            10 & $q=p^h$ is odd, $2e \mid  h$ & $n=\frac{t(q-1)}{p^e-1}$, $1 \leq t \leq p^e-1$ & $1 \leq k \leq \frac{(t-1)(q-1)}{p^{2e}} $ & \cite{liyang.gso}\\\hline

            11 & $q=p^h$ is odd, $2e \mid  h$ & $n=\frac{t(q-1)}{p^e-1}+2$, $1 \leq t \leq p^e-1$ &$k=\frac{t(q-1)}{p^{2e}-1}$ & \cite{liyang.gso}\\\hline

            12 & $q=p^h$ is odd, $2e \mid  h$ & $n=tp^{h-e}$, $1 \leq t \leq p^e$ &$1 \leq k \leq \lfloor \frac{p^e(tp^{h-2e}+1)-1}{p^e+1}\rfloor $& \cite{liyang.gso}\\\hline

            13 & $q=p^h$ is odd, $2e \mid  h$ & $\begin{array}{c} n=tp^{h-e}+1$, $1 \leq t \leq p^e,\\
                 (p^e+1) \mid  (tp^{h-2e}+1) \end{array}$ & $k= \frac{p^e(tp^{h-2e}+1)}{p^e+1} $ & \cite{liyang.gso}\\\hline

            14 & $q=p^h$ is odd, $2e \mid  h$ & $\begin{array}{c} n=r_1r_2, r_2=\frac{q-1}{\gcd(q-1,x_2)}, 1 \leq r_1 \leq \frac{q-1}{\gcd(q-1,x_1)},\\
              (q-1)\mid  {\lcm}(x_1,x_2),
               \gcd(x_2,q-1) \mid  x_1(p^e-1) \end{array}$ & $1 \leq k \leq \frac{p^e+r_2(r_1-1)}{p^e+1} $ & \cite{liyang.gso}\\\hline

            15 & $q=p^h$ is odd, $2e \mid  h$ & $\begin{array}{c} n=r_1\frac{q-1}{\gcd(q-1,x_2)}+2, 1 \leq r \leq \frac{q-1}{\gcd(q-1,x_1)},\\
                 p^e+1 \mid  n,
                 (q-1)\mid  {\lcm}(x_1,x_2),\\
                 \gcd(x_2,q-1) \mid  x_1(p^e-1) \end{array}$ & $k=\frac{n}{p^e+1}+1$ & \cite{liyang.gso}\\ \hline

            16 & $q=p^h$ is odd, $2e \mid  h$ & $\begin{array}{c} n=rm,1 \leq r \leq \frac{p^e-1}{m_1},\\
                m \mid  (q-1),m_1=\frac{m}{\gcd(m,y)},\\
                y=\frac{q-1}{p^e-1} \end{array}$ & $1 \leq k \leq \frac{p^e+m(r-1)}{p^e+1} $ & \cite{liyang.gso}\\\hline        
                               
            17 & $q=p^h$ is odd, $2e \mid  h$ & $\begin{array}{c} n=rm+2,1 \leq r \leq \frac{p^e-1}{m_1},\\
                     m \mid  (q-1),m_1=\frac{m}{\gcd(m,y)},\\
                     y=\frac{q-1}{p^e-1},(p^e+1) \mid  n \end{array}$ & $k=\frac{n}{p^e+1}+1  $ & \cite{liyang.gso}\\

            \hline 
            
            18 & $q=p^h$ is odd &$\begin{array}{c} n=wp^{mz} , 1 \leq w \leq p^m, 1 \leq z \leq \frac{h}{m}-1,\\
                2^t \mid \frac{h}{m}, 2^t=p^e+1 \end{array}$ & $1 \leq k \leq \lfloor \frac{n+p^e-1}{p^e+1}\rfloor $ & \cite{liyang.MDS} \\ \hline

            19 & $q=p^h$ is even & $n < q$, $\frac{h}{{\rm gcd}(e,h)}$ is odd and $h >1$ &  $1 \leq k \leq \lfloor \frac{n+p^e-1}{p^e+1}\rfloor $ & \cite{Q}\\\hline

            20 & $q=p^h>3$ & $n \leq r$, $r=p^\epsilon $ with $\epsilon \mid h$ and $(p^e+1) \mid \frac{q-1}{r-1}$ & $1 \leq k \leq \lfloor \frac{n+p^e-1}{p^e+1}\rfloor $ &\cite{Q}\\\hline

            21 & $q=p^h>3$ & $n \mid q$ &  $1 \leq k \leq \lfloor \frac{n+p^e-1}{p^e+1}\rfloor $ & \cite{Q}\\\hline

            22 & $q=p^h >3$ &$(n-1)\mid (q-1)$, $n >1$ and $p\mid n$ & $1 \leq k \leq \lfloor \frac{n+p^e-1}{p^e+1}\rfloor $ & \cite{Q}\\\hline
           
            23 & $q=p^h>3$ & $n=2n' \leq q$ and $n' \mid q$  &  $1 \leq k \leq \lfloor \frac{n+p^e-1}{p^e+1}\rfloor $ & \cite{Q}\\\hline
           
            24 & $q=p^h >3$ & $\begin{array}{c} n=n't \leq q, n'=r^\mu\  {\rm and}\  r=p^\epsilon\  {\rm with}\  \epsilon \mid h,\\
                {\rm where}\  1 \leq t \leq r, 1 \leq \mu \leq \frac{h}{\epsilon}, {\rm gcd}(p^e+1, q-1) \mid \frac{q-1}{r-1} \end{array}$& $1 \leq k \leq \lfloor \frac{n+p^e-1}{p^e+1}\rfloor $ & \cite{Q}\\

            \hline
        \end{tabular}}
    \end{center}
\end{table}

\begin{table}
    \caption{New (MDS) $e$-Galois SO AG codes} 
       \label{tab:2}       
       \begin{center}
        \resizebox{\linewidth}{!}{
           \begin{tabular}{ c c c c}
           
            \hline
            Finite field $\F_q$ & $n$ & Parameters & Reference \\\hline
            \hline

             $q=p^h$ is odd, $2e \mid h$ & $n=(t+1)\frac{q-1}{p^e+1}+1 $, $1 \leq t \leq p^e$, $(p^{2e}-1) \mid \frac{q-1}{p^e+1}$ & $\begin{array}{c}[n,k,n-k+1]_q, \\  1 \leq k \leq \lfloor \frac{(t+1)(q-1)+p^e(p^e+1)}{(p^e+1)^2}\rfloor\end{array}$ & Theorem \ref{t3}\\
             \hline
             
             $q=p^h$ is even, $2e \mid h$ & $\begin{array}{c} E=\{x^{2^e+1} \vert x \in \F_q^*\},\\ 1 \leq n \leq | U_1 |,
                U_1=\{\alpha \in E \vert Tr_{\F_q/\F_2}(\alpha ^3)=0\},\\
                {\rm or} \  1 \leq n \leq | U_2 |,
                U_2=\{\alpha \in E \vert Tr_{\F_q/\F_2}(\alpha +\alpha ^3)=0\}
            \end{array}$ & $\begin{array}{c}[2n,k-1, \geq 2n-k+1]_q, \\ 3 \leq k \leq \lfloor \frac{2n+p^e+1}{p^e+1} \rfloor\end{array}$  & Theorem \ref{t5}\\
                \hline
             $q=p^h$ is even, $h$ is even  & $(n-1) \mid (q-1)$ & $\begin{array}{c}[2n, k-\frac{\sqrt{q}}{2}, \geq 2n-k+1]_q, \\ 1+\sqrt{q} \leq k\leq \lfloor \frac{2n+p^e+\sqrt{q}-1}{p^e+1} \rfloor\end{array}$  & Theorem \ref{t6}\\
             \hline
             $q=p^h$ is odd, $2e \mid h$ & $ \begin{array}{c} n=tp^{aw}, a \mid e, 1 \leq t \leq p^a, 1 \leq w \leq \frac{h}{a}-1,\\
                {\rm or}\  n=tp^{h-e}, 1 \leq t \leq p^e,\\ {\rm or}\  n=(t+1)\frac{q-1}{p^e+1}+1, 1 \leq t \leq p^e, (p^{2e}-1) \mid \frac{q-1}{p^e+1} \\
                {\rm or}\  n=\frac{t(q-1)}{p^e-1}+a, a=\{0,1\}, 1 \leq t \leq p^e-1, \\
                {\rm or}\  n=\frac{r(q-1)}{\gcd(q-1,x_2)}+a, a=\{0,1\}, 1 \leq r \leq  \frac{q-1}{\gcd(q-1,x_1)}, \\ (q-1) \mid  {\lcm}(x_1,x_2), {\gcd}(x_2,q-1) \mid  x_1(p^e-1),\\
                {\rm or}\  n=rm+a,  a=\{0,1\}, m\mid (q-1), \\ 1 \leq r \leq \frac{p^e-1}{m_1}, m_1=\frac{m}{\gcd(m,y)}, y=\frac{q-1}{p^e-1} \end{array}$ 
                 & $\begin{array}{c}[\sqrt{q}n,k-\frac{q-\sqrt{q}}{2},\geq \sqrt{q}n-k+1]_q, \\ 1+q-\sqrt{q} \leq k \leq \lfloor \frac{\sqrt{q}n+p^e+q-\sqrt{q}-1}{p^e+1} \rfloor\end{array}$ & Theorem \ref{t7}\\

            \hline
        \end{tabular}}
    \end{center}
\end{table}

\section{Preliminaries}\label{sec2}
In this section, we review some basic notations and 
concepts on algebraic function fields and algebraic geometry codes.
\subsection{Algebraic function fields} 
In this subsection, we briefly introduce some concepts of algebraic function fields. 
For more details, the reader is referred to \cite{H.2009}. 

Let $\F_q(z)$ be the rational function field, where $z$ is a {\em transcendental element} over $\F_q$. 
Let $\mathcal{X}$ be a smooth projective curve with {\em genus} $g$ over $\F_q$. 
The field of rational functions of $\mathcal{X}$ is denoted by $\F_q(\mathcal{X})$. 
The field of functions of algebraic curves over $\F_q$ can be seen as a {\em finitely separable extension} of $\F_q(z)$. 
We use {\em places} $P$ of the function field $\F_q(z)$ to identify points on the curve $\mathcal{X}$.
Denote by $\mathbb{P}_{\F}$ the set of all places of $\F_q(z)$. Then
\begin{align*}
    \mathbb{P}_{\F}=\{P_{p(z)} | \  p(z)\ {\rm is\  a\  monic\  irreducible\  polynomial\ over }\ \F_q \} \cup P_{\infty},
\end{align*}
where $P_{p(z)}$ and $P_{\infty}$ were defined in \cite[Page 9]{H.2009}.
The {\em degree} of the place $P_{p(z)}$ is equal to the degree of $p(z)$ and the degree of $P_{\infty}$ is equal to 1 (see \cite[Definition 1.1.4]{H.2009}). 
If $P_{z(x)}=P_{z-\alpha}$, where $\alpha \in \F_q$, we call the places $P_{z(x)}$ and $P_{\infty}$ {\em rational places} of $\F_q(z)$. 
If all the coordinates of a point on $\mathcal{X}$ belong to $\F_q$, it is called a {\em rational point}. 
We denote the set of $\F_q$-rational points of $\mathcal{X}$ by $N(\mathcal{X})$.

Let $\mathbb{Z}$ be the {\em integer ring}. Then a {\em divisor} $G$ of $\F_q(\mathcal{X})$ can be defined as follows. 
\begin{definition}{\rm(\!\!\cite[Definition 1.1.4]{H.2009})}.
    A {\em divisor} $G$ of $\F_q(\mathcal{X})$ is a formal sum $G=\sum _{P \in \mathcal{X}}n_P P$ with $n_P \in \mathbb{Z}$ 
    and almost all $n_P=0$. The {\em support} and {\em degree} of $G$ are respectively given by  
    \begin{equation}
        \supp(G)=\{P \in \mathcal{X}\mid  n_P \neq 0\}~{\rm and}~\deg(G)= \sum _{P \in \mathcal{X}}n_P \deg(P).
    \end{equation}
\end{definition}

\begin{definition}{\rm(\!\!\cite[Definition 1.1.4]{H.2009})}.
    For two divisors $G=\sum_{P \in \mathcal{X}}n_P P$ and $H=\sum_{P \in \mathcal{X}}m_P P$, 
    we have $G \leq H$ if and only if $n_P \leq m_P$ for all places $P \in \mathcal{X}$.
\end{definition}

For every non-zero rational function $z \in \F_q(\mathcal X)$, the {\em principal divisor} of $z$ is 
    $(z)=\sum_{P \in \mathbb{P}_{\F}}v_P(z)P,$
where $v_P(\cdot )$ is the {\em discrete valuation} at the place $P$. 
We denote the set of zeros (resp. poles) of $(z)$ by $S$ (resp. $R$), then $(z)$ can be uniquely written as $(z)=\sum_{P \in S}v_P(z)P-\sum_{Q \in R}v_Q(z)Q$, where $v_P(z) >0$ for any $P \in S$ and $v_Q(z)>0$ for any $Q \in R$.
The divisor $\sum_{P \in S}v_P(z)P$ is referred to as the {\em zero divisor} of $(z)$, denoted by $(z)_0$.
The divisor $\sum_{Q \in R}v_Q(z)Q$ is referred to as the {\em pole divisor} of $(z)$, denoted by $(z)_{\infty}$.
As a result, every principal divisor can be uniquely expressed as $(z)=(z)_0-(z)_{\infty}$.
It can also be seen that $\deg((z)_0)=\deg((z)_{\infty})$ (see \cite[Theorem 1.4.11]{H.2009}),
then all degree of principal divisors are equal to zero.

For two divisors $G$ and $H$ on the curve $\mathcal{X}$, if $G=H+(z)$ for some rational function $z \in \F_q(\mathcal X)$,  
they are said to be {\em equivalent}.

For a divisor $G$ on the curve $\mathcal{X}$, the {\em Riemann-Roch space associated to $G$}, denoted by $\mathcal{L}(G)$, is defined by 
\begin{equation}
    \mathcal{L}(G)=\{z \in \F_q(\mathcal{X}) \setminus \{0\}\mid (z)+G \geq 0\} \cup \{0\}.
\end{equation} 
Clearly, $\mathcal{L}(G)$ is a finite dimensional vector space over $\F_q$  
and we denote $\ell(G)=\dim(\mathcal{L} (G))$, where $\dim(\mathcal{L} (G))$ is the dimension of $\mathcal{L} (G)$. 
Then $\ell(G)=0$ if $\deg(G) < 0$.
Let $w=zdx$ be a {\em Weil differential} of $\F_q(\mathcal{X})$, where $z \in \F_q(\mathcal{X})$. 
We call $(w)=\sum_{P \in \mathcal{X}}v_P(w) P$ a {\em canonical divisor} if $w\neq 0$ and denote the canonical divisor by $W=(w)$.
From \cite[Corollary 1.5.16]{H.2009}, we have $\deg(W)=2g-2$, where $g$ is the genus of the curve $\mathcal{X}$. 
Now, we can use the following well-known Riemann-Roch's theorem to calculate $\ell(G)$.
\begin{lem}{\rm (Riemann-Roch's theorem \cite[Theorem 1.5.15]{H.2009})}.
    Let $W$ be a canonical divisor of $\F_q(\mathcal{X})$. Then for each divisor $G$, we have 
    \begin{equation*}
        \ell(G)=\deg(G)+1-g+\ell(W-G), 
    \end{equation*}
    where $g$ is the genus of $\mathcal{X}$. 
\end{lem}

\subsection{Algebraic geometry codes} 

Let $\Omega = \{zdx\mid z \in \F_q(\mathcal{X})\}$ be the set of all Weil differentials of $\F_q(\mathcal{X})$.
For a divisor $G$, we define 
\begin{equation}
    \Omega (G)=\{w \in \Omega  \backslash  \{0\}\mid (w)-G \geq 0\} \cup \{0\}.
\end{equation}
Then it can be shown that $\Omega(G)$ is a finite dimensional vector space 
and we can define the {\em algebraic geometry (AG) code} (see \cite[Definition 2.2.1]{H.2009}) and the {\em differential geometry (DG) code} (see \cite[Definition 2.2.6]{H.2009})  
as follows. 
\begin{definition}\label{d2}
    Let $\mathcal{X}$ be a smooth projective curve.
    Let $P_1, P_2, \ldots , P_n$ be $n$ pairwise distinct rational points on $\mathcal{X}$.  
    Suppose that $G$ and $D=P_1+P_2+\cdots+P_n$ are two divisors such that $\supp(D) \cap \supp(G) = \emptyset $. 
    \begin{enumerate}
        \item[\rm(1)] The {\em algebraic geometry (AG) code} $C_{\mathcal{L}}(D,G)$ associated with the divisors $D$ and $G$ is the image of the linear map 
        $ev_{\mathcal{L}}: \mathcal{L}(G) \to \F_q^n$, $z \mapsto (z(P_1), z(P_2),\ldots,z(P_n))$, i.e., 
        \begin{align}
            C_{\mathcal{L}}(D,G)=\{(z(P_1), z(P_2), \ldots,z(P_n))\mid z \in \mathcal{L}(G)\}.
        \end{align}
        Moreover, for any $\bm{v}=(v_1,v_2,\ldots,v_n)$, where $v_i\in \F_q^*$ for each $1\leq i\leq n$, 
        the code $C_{\mathcal{L}}(D,G,\bm{v})$ is called a {\em generalized algebraic geometry code}, defined by
        \begin{align}
            C_{\mathcal{L}}(D,G,\bm{v})=\{(v_1z(P_1), v_2z(P_2), \ldots,v_nz(P_n))\mid z \in \mathcal{L}(G)\}.
        \end{align}

        \item[\rm(2)] The {\em differential geometry (DG) code} $C_{\Omega }(D,G)$ associated with $D$ and $G$ is the image of the linear map
        $ev_{\Omega}: \Omega(G-D) \to \F_q^n$, $w \mapsto ({\Res}_{P_1}(w), {\Res}_{P_2}(w), \ldots ,{\Res}_{P_n}(w))$, i.e., 
        \begin{align}
            C_{\Omega }(D,G)=\{({\Res}_{P_1}(w),{\Res}_{P_2}(w),\ldots,{\Res}_{P_n}(w))\mid w \in \Omega(G-D) \},
        \end{align}
        where $\Res_{P_i}(w) $ is the residue of $w$ at the point $P_i$.
    \end{enumerate}
\end{definition}
For convenience, we refer to both $C_{\mathcal{L}}(D,G)$ and $C_{\mathcal{L}}(D,G,\bm{v})$ as AG codes in the sequel. 
Then the following lemma gives the parameters of an AG code $C_{\mathcal{L}}(D,G,\bm{v})$.

\begin{lem}{\rm (\!\!\cite[Theorem 2.2.2]{H.2009})}.\label{p.parameters}
   Let $\mathcal{X}$ be a smooth projective curve over $\F_q$ with genus $g$. 
   Then the AG code $C_{\mathcal{L}}(D,G,\bm{v})$ is an $[n,k',d]$ code, 
   where  $k'=\ell(G)-\ell(G-D)$ and $d \geq n-\deg(G)$.
   Moreover, if $2g-2 < \deg(G) < n$, then $k'=\deg(G)+1-g$.
\end{lem}

\begin{lem}{\rm (\!\!\cite[Proposition 2.1.10]{H.2009})}.\label{ralation}
    With notations as above,  
    let $C_{\mathcal{L}}(D,G,\bm{v})$ and $C_{\mathcal{L}}(D,H,\bm{v})$ be two AG codes satisfying $G \leq H$. 
    Then $C_{\mathcal{L}}(D,G, \bm{v}) \subseteq C_{\mathcal{L}}(D,H,\bm{v})$.
\end{lem}

\begin{lem}{\rm (\!\!\cite[Proposition 2.2.10]{H.2009})}.\label{p0}
    Let notations be the same as above.
    Let $w$ be a Weil differential such that $v_{P_i}(w)=-1$ and ${\Res}_{P_i}(w) \neq 0$ for each $1\leq i\leq n$. 
    Suppose that $\bm{v}=(v_1,v_2,\ldots ,v_n)\in (\F_q^*)^n$. Then we have  
    \begin{align}
        {C_{\mathcal{L}}(D,G,\bm{v})}^{\perp_0} = \bm{v}^{-1} C_{\Omega}(D,G)= \bm{v}^{-1} \bm{\Res_P(w)} C_{\mathcal{L}}(D,H),
    \end{align}
    where $H=D-G+(w)$, $\bm{v}^{-1}=(v_1^{-1}, v_2^{-1}, \ldots, v_n^{-1})$ and $\bm{\Res_P(w)}=({\Res_{P_1}}(w),\Res_{P_2}(w), \ldots,{\Res_{P_n}}(w))$.
\end{lem}

Put $\bm{v}=(v_1,v_2,\ldots,v_n)$, where $v_i\in \F_q^*$ for each $1\leq i\leq n$ 
and ${\bm{\alpha}} =(\alpha _1,\alpha _2,\ldots,\alpha_n)$, where $\alpha_i,\alpha_j\in \F_q$ 
satisfying $\alpha_i\neq \alpha_j$ for each $1\leq i \neq j\leq n$. 
Let $h(x)=\prod _{i=1}^n(x-\alpha _i)$ and $P_i=P_{x-\alpha _i}$ be the rational points 
corresponding to the irreducible polynomial $x-\alpha _i$ for each $1\leq i \leq n$.
The {\em generalized Reed-Solomon (GRS) code} $GRS_k(\bm{\alpha},\bm{v})$ associated to $\bm{\alpha}$ and $\bm{v}$ 
is defined by 
\begin{equation}\label{eq.GRS}
    GRS_k(\bm{\alpha},\bm{v})=\{(v_1f(\alpha _1), v_2f(\alpha _2),\ldots,v_nf(\alpha _n))\mid f(x) \in \F_q(x),\deg(f(x)) \leq k-1\}. 
\end{equation}
Obviously, $GRS_k(\bm{\alpha} ,\bm{v})$ defined in Equation (\ref{eq.GRS}) is an $[n,k,n-k+1]_q$ MDS code. 
\begin{lem}{\rm (\!\!\cite[Proposition 2.3.3]{H.2009})}.\label{lem.generator matrix}
    Let $\mathcal{X}$ be a smooth projective line and other notations be the same as before. 
    Suppose that $D=P_1+P_2+\cdots+P_n$ and $G=(k-1)P_\infty$ such that $\supp(D) \cap \supp(G) = \emptyset $. 
    Then the AG code $C_{\mathcal{L}}(D,G,\bm{v})$ is equal to the GRS code $GRS_k(\bm{\alpha},\bm{v})$. 
    Moreover, a generator matrix of $C_{\mathcal{L}}(D,G,\bm{v})$ is given by 
    \begin{equation}
        \left[\begin{array}{cccc}
            v_1 & v_2 & \cdots & v_n \\
            v_1\alpha _1 & v_2\alpha _2 & \cdots & v_n\alpha _n \\
            \vdots & \vdots & \ddots & \vdots\\
            v_1\alpha _1^{k-1} & v_2\alpha _2^{k-1} & \cdots & v_n\alpha _n^{k-1} \\
        \end{array}\right].
    \end{equation}
\end{lem}

\subsection{Some auxiliary results}
Let notations be the same as before. 
For any vector ${\bm{c}}=(c_1,c_2,\ldots,c_n)\in \F_q^n$, we denote ${\bm{c}^{p^e}}=(c_1^{p^e},c_2^{p^e},\ldots,c_n^{p^e})$. 
Then for a linear code $\C$, we denote 
$\C^{p^e}=\left\{\mathbf{c}^{p^e}\mid \mathbf{c}\in \C\right\}.$

\begin{lem}{\rm(\!\!\cite[Lemma III.2]{caomeng})}\label{l2}
    Let $q=p^h$ and $0 \leq e \leq h-1$.  Let $E=\{x^{p^e+1}\mid x \in \F_q^*\}$.
    Then $\F_{p^e}^*\subseteq E$ if and only if $2e \mid  h$.  
\end{lem}

\begin{lem}[\!\!\cite{Gd2,Gd1}]\label{p1}
    Let $q=p^h$ and  $\C$ be a linear $[n,k,d]_q$ code.  
    Then for each $0\leq e\leq h-1$, we have 
    \begin{equation}
        \C^{\perp _e}=(\C^{p^{h-e}})^{\perp_0}=(\C^{\perp_0})^{p^{h-e}}.
    \end{equation}
\end{lem}

\begin{lem}{\rm(\!\!\cite[Hilbert's Theorem 90]{32 th2.25})}\label{l4}
    Let $q=p^h$ and $\alpha \in \F_q$. Then the equation $y^p-y=\alpha $ have solutions over $\F_q$ if and only if $Tr_{\F_q/\F_p}(\alpha )=0$, 
    where $Tr_{\F_q/\F_p}$ is the trace mapping from $\F_q$ to $\F_p$.
\end{lem}

\begin{lem}{\rm(\!\!\cite[Lemma 10]{sok.2021 sd})}\label{l5}
    Let $q=2^h$, where $h$ is even. If $\alpha \in \F_{\sqrt{q}}$, then $Tr_{\F_q / \F_2}(\alpha )=0$ and $Tr_{\F_q / \F_2}(\alpha^3+\alpha )=0$
\end{lem}

\section{Galois SO codes from AG codes}\label{sec 30}

In this section, we present a criterion for an AG code being Galois SO.
Then, we propose the embedding method and show a concrete construction of MDS Galois SO AG codes over projective lines. In addition, we
also construct some new Galois SO AG codes over projective elliptic curves, hyper-elliptic
curves, and Hermitian curves.

We now begin this section with the following important lemma. 
This lemma gives a criterion for an AG code being Galois SO 
and hence, it gives an affirmative answer for the first part of Problem \ref{prob.1}.

\begin{lem}\label{t1}
    Let $q=p^h$ and $0 \leq e \leq h-1$.
    Let $\mathcal{X}$ be a smooth projective curve with genus $g$. 
    Let $D=P_1+P_2+\cdots+P_n$ and $G=(k-1)P_\infty$ be two divisors such that $\supp(D) \cap \supp(G) = \emptyset $, 
    where $P_1, P_2, \ldots , P_n$ are $n$ pairwise distinct rational points on $\mathcal{X}$ and $2g-2 < \deg(G) < n$. 
    Let $w$ be a Weil differential such that $v_{P_i}(w)=-1$ for each $1\leq i\leq n$ and $H=D-G+(w)$ satisfying $\supp(D)\bigcap \supp(H)=\emptyset$.
    For any $\bm{v}=(v_1,v_2,\ldots,v_n)\in (\F_q^*)^{n}$, if  $p^eG\leq H$ and ${\Res}_{P_i}(w)=v_i^{p^e+1}$ for each $1\leq i\leq n$,
    then $C_\mathcal{L} (D,G,\bm{v})$ is an $[n,k-g,d]_q$ $e$-Galois SO AG code, 
    where $2g+1\leq k \leq \lfloor \frac{n+p^e+2g-1}{p^e+1} \rfloor$ and $d\geq n-\deg(G)$. 
\end{lem}
\begin{proof}
    Under the given conditions, it can be calculated that 
    \begin{align*}
        ({C_\mathcal{L} (D,G,\bm{v})}^{\perp_e})^{p^e} & = (C_\mathcal{L} (D,G,\bm{v}))^{\perp_0}~({\rm By~ Lemma}~\ref{p1}) \\ 
                                                      & = \bm{v}^{-1}\bm{{\Res}_{P}(w)} C_\mathcal{L} (D,H)~({\rm By~ Lemma}~\ref{p0}) \\ 
                                                      & = C_\mathcal{L} (D,H,\bm{v}^{p^e}),
    \end{align*} 
      where $\bm{{\Res}_P(w)}=(v_1^{p^e+1},v_2^{p^e+1},\ldots,v_n^{p^e+1})$.
    By Definition \ref{d2}, we have
    $(C_\mathcal{L} (D,G,\bm{v}))^{p^e}=({v}_1^{p^e}z^{p^e}(P_1), {v}_2^{p^e}z^{p^e}(P_2), \ldots,{v}_n^{p^e}z^{p^e}(P_n))$, where $z \in \mathcal{L}(G)$. 
    It follows from $z^{p^e} \in \mathcal{L}(p^eG)$ that $(C_\mathcal{L} (D,G,\bm{v}))^{p^e} \subseteq C_\mathcal{L} (D,p^eG,\bm{v}^{p^e})$.
    Since  $p^eG \leq H$, by Lemma \ref{ralation}, 
    we further have $C_\mathcal{L} (D,p^eG,\bm{v}^{p^e})\subseteq C_\mathcal{L} (D,H,\bm{v}^{p^e})$.
    It implies that $(C_\mathcal{L} (D,G,\bm{v}))^{p^e}\subseteq ((C_\mathcal{L} (D,G,\bm{v}))^{\bot_e})^{p^e}.$ 
    Then $C_\mathcal{L} (D,G,\bm{v})\subseteq (C_\mathcal{L} (D,G,\bm{v}))^{\bot_e}.$ 
    Note that the parameters of $C_\mathcal{L} (D,G,\bm{v})$ follows immediately from Lemma \ref{p.parameters}. 
    This completes the proof. 
\end{proof}

\begin{remark}  
    We emphasize here that Lemma \ref{t1} and \cite[Lemma 4]{sok.2021 hso} are actually different.
    Specifically speaking, in \cite[Lemma 4]{sok.2021 hso}, 
    the author constructed some Hermitian SO AG codes $C_{\mathcal{L}_q}(D,G,\bm{v})$ with the space  $\mathcal{L}_q(G)=\{z\in \F_q(\mathcal{X})\backslash \{0\}\mid z(P_i)\in \F_{\sqrt{q}},\ {\rm{for}} \ 1 \leq i \leq n, (z)+G \geq 0\}\cup \{0\}$.
    However, Lemma \ref{t1} takes the space $\mathcal{L}(G)=\{z \in \F_q(\mathcal{X}) \backslash  \{0\}\mid (z)+G \geq 0\} \cup \{0\}$.
    It can be checked that  $\mathcal{L}_q(G) \subseteq \mathcal{L}(G)$.
    Hence, when $e=\frac{h}{2}$ and $h$ is even, there may exist more new Hermitian SO AG codes $C_\mathcal{L} (D,G,\bm{v})$ constructed by Lemma \ref{t1} compared to \cite[Lemma 4]{sok.2021 hso}.
        
\end{remark}

\subsection{MDS Galois SO AG codes from projective lines}\label{sec3}
In this subsection, we first propose an embedding method to construct new MDS Galois SO codes 
from MDS Galois SO AG codes constructed over projective lines.
Then we give a specific construction of MDS Galois SO AG codes. 

\begin{lem}\label{l1}
   Let $q=p^h$, where $p$ is an odd prime. 
   Let $\mathcal{X}$ be a smooth projective line and $P_1, P_2, \ldots , P_n$ be $n$ pairwise 
   distinct rational points on $\mathcal{X}$. 
   Let $G=(k-1)P_\infty$ and $D=P_1+P_2+\cdots+P_n$ be two divisors such that $\supp(G)\bigcap \supp(D)=\emptyset$. 
   Let $w$ be a Weil differential such that ${\Res}_{P_i}(w)={v_i}^{p^e+1}$ 
   for some $v_i \in \F_q^*$, where $1\leq i \leq n$. 
   Suppose that $C_\mathcal{L} (D,G,\bm{v})$ is an $[n,k,n-k+1]_q$ MDS $e$-Galois SO AG code, 
   where $1\leq k \leq\lfloor \frac{n+p^e-1}{p^e+1}\rfloor$, $0\leq e\leq h-1$ and $\bm{v}=(v_1,v_2,\ldots,v_n) \in ({\F_q^*})^n$. 
   Then the following statements hold. 
   \begin{enumerate}
        \item[\rm(1)] If $k=\frac{n-1}{p^e+1}$, then $C_\mathcal{L} (D,G,\bm{v})$ can be embedded 
        into an $[n+1,k+1,n-k+1]_q$ MDS $e$-Galois SO code.

        \item[\rm(2)] If $k<\frac{n-1}{p^e+1}$, then $C_\mathcal{L} (D,G,\bm{v})$ can be embedded into an $[n,k+1,n-k]_q$ MDS $e$-Galois SO code.

        \item[\rm(3)] If $k>\frac{n-1}{p^e+1}$ and $(n-1) \mid  (q-1)$, then $C_\mathcal{L} (D,G,\bm{v})$ 
        can be embedded into an $[n,k+1,n-k]_q$ MDS $e$-Galois SO code.
   \end{enumerate}
\end{lem}
\begin{proof}
    By Lemma \ref{lem.generator matrix}, we can write the generator matrix of  $C_\mathcal{L} (D,G,\bm{v})$ as 
    \begin{equation*}
        G=\left[\begin{array}{cccc}
            v_1 & v_2 & \cdots & v_n \\
            v_1\alpha _1 & v_2\alpha _2 & \cdots & v_n\alpha _n \\
            \vdots & \vdots & \ddots & \vdots\\
            v_1\alpha _1^{k-1} & v_2\alpha _2^{k-1} & \cdots & v_n\alpha _n^{k-1} \\
        \end{array}\right] =\left[\begin{array}{c}
            {\bm g_1}\\
            {\bm g_2}\\
            \vdots\\
            {\bm g_k}\\
        \end{array}\right].
    \end{equation*}
    Consider the system of equations 
    \begin{align}\label{eq.system}
       B_{(p^e+1)(k-1)}X^T = \bm{0},
    \end{align}
    where $B_j=\left[\begin{array}{cccc}
        1 & 1 & \cdots & 1\\
        \alpha_1  & \alpha_2  & \cdots & \alpha_n \\
        \vdots & \vdots & \ddots & \vdots \\
        \alpha_1^j  & \alpha_2^j  & \cdots & \alpha_n^j \\
    \end{array}\right]$ for $(k-1)\leq j\leq (p^e+1)(k-1)$, $X^T$ is the transpose of $X$ and $\bm{0}=(0,0,\ldots,0)^T$ is a column zero vector of length $(p^e+1)(k-1)$. 
    Since $C_\mathcal{L} (D,G,\bm{v})$ is $e$-Galois SO, 
    $\bm{v}^{p^e+1}=(v_1^{p^e+1},v_2^{p^e+1},\ldots,v_n^{p^e+1})$ is always a nonzero solution of Equation (\ref{eq.system}). 

    Let $\gamma \in \F_q$. 
    Next we consider the linear code $\widetilde{C_\mathcal{L}(D,G,\bm{v})}$ with a generator matrix 
    \begin{equation*}
        \widetilde{G}=\left[\begin{array}{ccccc}
        v_1 & v_2 & \cdots & v_n & 0\\
        v_1 \alpha _1 & v_2 \alpha _2 & \cdots & v_n \alpha _n & 0\\
        \vdots & \vdots & \ddots & \vdots & \vdots\\
        v_1 \alpha_1^{k-1} & v_2 \alpha _2^{k-1} & \cdots & v_n \alpha _n^{k-1} & 0\\
        v_1 \alpha_1^{k} & v_2 \alpha _2^{k} & \cdots & v_n \alpha _n^{k} & \gamma \\
        \end{array}\right] =\left[\begin{array}{c}
            \widetilde{\bm g_1}\\
            \widetilde{\bm g_2}\\
            \vdots\\
            \widetilde{\bm g_k}\\
            \widetilde{\bm g_{k+1}}\\
        \end{array}\right]
    \end{equation*}
    and the system of equations 
    \begin{align}\label{(1)}
        B_{(p^e+1)k}X^T = \bm{\varepsilon},
    \end{align}
    where $\bm{\varepsilon} =(0,0,\ldots,0,-\gamma^{p^e+1})^T $ is a column vector of length $(p^e+1)k$. 
    For our desired results, we need to determine whether Equation (\ref{(1)}) has nonzero solutions.  
    We have the following three cases. 

    \textbf{Case 1:}
    If $k(p^e+1)+1=n$, then $B_{(p^e+1)k}$ is invertible. 
    Put $\gamma \in \F_q^*$ satisfying $\gamma^{p^e+1}=-1$.
    By Cramer's rule, Equation (\ref{(1)}) has the same (unique) solution 
    $\bm{v}^{p^e+1}=(v_1^{p^e+1},v_2^{p^e+1},\ldots,v_n^{p^e+1})$ with Equation (\ref{eq.system}), 
    where $v_i^{p^e+1}=\prod_{j=1,j\neq i}(\alpha _i-\alpha _j)^{-1}$ for $1 \leq i \leq n $.
    It implies that $\langle \widetilde{\bm g_i}, \widetilde{\bm g_{k+1}} \rangle _e=0$ for each $1 \leq i \leq k+1$. 
    On one hand, it is obvious that the subcode ${\widetilde{C_\mathcal{L}(D,G,\bm{v})}}_{(k)}$ of $\widetilde{C_\mathcal{L}(D,G,\bm{v})}$,   
    generated by $\widetilde{\bm g_1}, \widetilde{\bm g_2}, \ldots, \widetilde{\bm g_k}$, is an $[n+1,k]_q$ $e$-Galois SO code 
    since $C_\mathcal{L}(D,G,\bm{v})$ is $e$-Galois SO. 
    On the other hand, $\widetilde{\bm g_1}, \widetilde{\bm g_2}, \ldots, \widetilde{\bm g_k}, \widetilde{\bm g_{k+1}}$     
    is linear independent. 
    Hence, $\widetilde{C_\mathcal{L}(D,G,\bm{v})}$ is an $[n+1,k+1,n-k+1]_q$ MDS $e$-Galois SO code.
    
    \textbf{Case 2:} If $k(p^e+1)+1 < n$, it is clear that Equation (\ref{(1)}) only has a trivial solution if $\gamma\neq 0$. 
    Then by a similar argument to \textbf{Case 1} above, deleting the last zero coordinate of $\widetilde{C_\mathcal{L}(D,G,\bm{v})}$ 
    yields an $[n,k+1,n-k]_q$ MDS $e$-Galois SO code. 

    \textbf{Case 3:} If $k(p^e+1)+1>n$, it is easily seen that the last $(p^e+1)k-n+2$ rows of $B_{(p^e+1)k}$ are 
    linearly dependent on the first $n-1$ rows since $(n-1) \mid  (q-1)$. 
    Hence, the number of linearly independent equations in Equation (\ref{(1)}) is less than $n$. 
    Then \textbf{Case 3} is similar to \textbf{Case 2}. 

    We complete the proof. 
\end{proof}
 


In the following, we give an explicit constructions of MDS Galois SO codes from projective lines. 
To this end, we need the following lemma.
\begin{lem}\label{l3}
    Let $q=p^h$ and  $2e \mid  h$, where $p$ is an odd prime and $0\leq e\leq h-1$. 
    Let $\alpha_i$ and  $\alpha _j$ be two distinct elements of $\F_q^*$. 
     Assume that $E=\{x^{p^e+1}\mid x \in \F_q^*\}$ and $(p^{2e}-1) \mid \frac{q-1}{p^e+1}$, 
    then $\alpha_i^{\frac{q-1}{p^e+1}}-\alpha _j^{\frac{q-1}{p^e+1}} \in E$.
\end{lem}

\begin{proof}
    Denote $N=\frac{q-1}{p^e+1}$.
    Let $w$ be a primitive element of $\F_q$. 
    Then we can write $\alpha_i=w^a$ and  $\alpha _j=w^b$, 
    where $a$ and $b$ are distinct modulo $(q-1)$. 
    Raising $\alpha_i^{N}-\alpha _j^{N}$ to the power $p^e+1$ , we have
    \begin{align*}
        (\alpha_i^{N}-\alpha _j^{N})^{p^e-1}&=\frac{\alpha _i^{Np^e}-\alpha_j^{Np^e}}{\alpha _i^N-\alpha _j^N}=\frac{\alpha _i^{q-1-N}-\alpha_j^{q-1-N}}{\alpha _i^N-\alpha _j^N}
        =\frac{w^{\frac{q-1}{2}}}{w^{aN}w^{bN}}=w^{\frac{q-1}{2}-(a+b)N}
    \end{align*}
    where $w^{\frac{q-1}{2}}=-1$ since $q$ is odd.
    It further deduces from $(p^{2e}-1) \mid N$ and $(p^{2e}-1) \mid (q-1)$ that 
    there is an integer $c$ such that $\alpha_i^{N}-\alpha _j^{N}=w^{\frac{q-1}{2(p^e-1)}-\frac{(a+b)N}{p^e-1}+\frac{c(q-1)}{p^e-1}}=w^{\left( \frac{q-1}{2(p^{2e}-1)}-\frac{(a+b)N}{p^{2e}-1}+\frac{c(q-1)}{p^{2e}-1}\right)(p^e+1)} \in E$.
    This completes the proof.
\end{proof}

\begin{theorem}\label{t3}
    Let $q=p^h$ and $2e \mid  h$, where $p$ is an odd prime and $0\leq e\leq h-1$. 
    Let $n=(t+1)\frac{q-1}{p^e+1}+1$ and $(p^{2e}-1) \mid \frac{q-1}{p^e+1}$, where  $1 \leq t \leq p^e$. 
    Then for any $1 \leq k \leq \lfloor \frac{(t+1)(q-1)+p^e(p^e+1)}{(p^e+1)^2}\rfloor$, 
    there exists an $[n,k,n-k+1]_q$ MDS $e$-Galois SO AG code.
\end{theorem}

\begin{proof}
    Denote $U_N=\{\alpha\in \F_q\mid \alpha ^N=1\}=\{u_1,u_2,\ldots,u_N\}$, where $N=\frac{q-1}{p^e+1}$. 
    Let $w$ is a primitive element of $\F_q$.
    Assume that there exist $t$ distinct multiplicative cosets with their coset leaders 
    $\alpha _1=w^{\lambda _1},\alpha _2=w^{\lambda_2}, \ldots, \alpha _t=w^{\lambda_t}$ being elements in $\F_q$, where $\lambda_1,\lambda _2, \ldots,\lambda _t$ is distinct modulo $p^e+1$.
    Then $\alpha _1U_N,\alpha _2U_N, \ldots, \alpha _tU_N$ are $t$ distinct nonzero cosets of $U_n$.

    Put $U=U_N \cup (\cup_{i=1}^t\alpha _iU_N)\cup\{0\}$ and let 
    \begin{equation*}
        h(x) =\prod _{\gamma \in U}(x- \gamma ) = (x^{N+1}-x)\prod _{i=1}^{t}(x^N-\alpha _i^N).
    \end{equation*}
    Then the derivative $h'(x)$ of $h(x)$ is given by
    \begin{equation*}
        h'(x)=((N+1)x^N-1)\prod _{i=1}^{t}(x^N-\alpha _i^N) + Nx^N(x^N-1)\sum_{j=1}^{t} \prod _{i=1, i\neq j}^t(x^N-\alpha _i^N). 
    \end{equation*}
    Let $E=\{x^{p^e+1} \mid x \in \F_q^* \}$. 
    For any $1 \leq j \leq t$ and $1 \leq s \leq N$, 
    we have
    \begin{equation*}
        h'(u_s) = N\prod _{i=1}^{t} (1-\alpha _i^N),
    \end{equation*}
    
    \begin{equation*}
        h'(\alpha _ju_s) =N\alpha _j^N(\alpha _j^N-1)\prod _{i=1,i \neq j}^t(\alpha _j^N-\alpha _i^N),
    \end{equation*}
    and
    \begin{equation*}
        h'(0)=(-1)^{t+1}\prod_{i=1}^t\alpha _i^N=w^{\frac{(t+1)(q-1)}{2(p^e+1)}(p^e+1)}\prod_{i=1}^t\alpha _i^N.
    \end{equation*}
    It turns out from 
    Lemmas \ref{l2} and \ref{l3} that $h'(u_s)$, $h'(\alpha _ju_s)$ and $h'(0)$ are contained in $E$ since $(p^{2e}-1) \mid N$.
    Hence, there is $\beta_j \in \F_q^*$ such that $h'(\gamma_j ) = \beta_j ^{p^e+1}$ for any $\gamma_j  \in U$, where $1 \leq j \leq n$. 

    Set $G=(k-1)P_\infty$, $D=(h(x))_0$ and $w=\frac{dx}{h(x)}$. 
    Then $H=D-G+(w) = ((t+1)\frac{q-1}{p^e+1} +1-2)P_\infty-G$ and 
    ${\Res}_{P_j}(w)=\frac{1}{h'(\gamma _j)}=\frac{1}{\beta _j^{p^e+1}} $, where $1 \leq j \leq n$.     
    By taking $\bm{v}=(\beta _1^{-1}, \beta _2^{-1},\ldots, \beta _n^{-1})$, 
    it follows that the length of the AG code $C_\mathcal{L} (D,G,\bm{v})$ is $(t+1)\frac{q-1}{p^e+1} +1=n$ 
    and the dimension of the AG code $C_\mathcal{L} (D,G,\bm{v})$ is $\deg(G)+1=k$.     
    Since $1 \leq k \leq \lfloor \frac{(t+1)(q-1)+p^e(p^e+1)}{(p^e+1)^2}\rfloor $, we have $p^eG \leq H$. 
    Then from Lemma \ref{t1}, the AG code $C_\mathcal{L} (D,G,\bm{v})$ is an $[n,k,n-k+1]_q$ MDS $e$-Galois SO code. 

    We have completed the proof. 
\end{proof}

Applying Lemmas \ref{l1} (1) and (3) to Theorem \ref{t3}, one can immediately obtain more new MDS Galois SO codes in the following corollary. 
Note that Part (2) of Lemma \ref{l1} can not yield new MDS Galois SO codes in this specific case 
because of the flexible range of available dimensions in Theorem \ref{t3}.

\begin{cor}\label{c1}
    Let notations be the same as in Theorem \ref{t3}. 
    Then the following statements hold.
    \begin{enumerate}
        \item[\rm(1)] If $k=\frac{(t+1)(q-1)}{(p^e+1)^2}$, there exists an $[n+1,k+1,n-k+1]_q$ MDS $e$-Galois SO code.              
    
        \item[\rm(2)] If $k > \frac{(t+1)(q-1)}{(p^e+1)^2}$ and
                     $ \frac{(t+1)(q-1)}{p^e+1} \mid (q-1)$, 
        there exists an $[n,k+1,n-k]_q$ MDS $e$-Galois SO code.    
    \end{enumerate}
\end{cor}

\begin{remark}
    Compared with Table \ref{tab:1}, it is easily seen that the construction presented in Theorem \ref{t3} is new 
    since it produces some MDS Galois SO codes with a different length form. 
    Hence, Corollary \ref{c1} will also yield some new MDS Galois SO codes. 
\end{remark}

A specific example of MDS Galois SO codes 
from Theorem \ref{t3} as well as Corollary \ref{c1} are given in the following. 


\begin{example}
      Let notations be the same as in Theorem \ref{t3}.
       Take $p=3$ and $h=8$, then the case satisfying $2e \mid h$ and $(p^{2e}-1) \mid N$ is only $e=1$. 
       Hence, we list some MDS $1$-Galois SO codes as follows.
       \begin{itemize}
          
        \item [\rm (1)] From Theorem \ref{t3}, we obtain 
        $[1641,k,1642-k]_{3^8}\ (1 \leq k \leq 410)$,
        $[3281,k,3282-k]_{3^8}\ (1 \leq k \leq 820)$,
        $[4921,k,4922-k]_{3^8}\ (1 \leq k \leq 1230)$,
        $[6561,k,6562-k]_{3^8}\ (1 \leq k \leq 1640)$ MDS $1$-Galois SO AG codes. 
   
        \item [\rm (2)] From Corollary \ref{c1}, we obtain 
        $[1642,411,1232]_{3^8}$, $[3282,821,2462]_{3^8}$, $[4922,1231,3692]_{3^8}$, $[6562,1641,4922]_{3^8}$ MDS $1$-Galois SO codes. 
     \end{itemize}
\end{example}

\subsection{Galois SO AG codes from projective curves}\label{sec4}

In this subsection, 
we will employ some special projective curves to construct Galois SO AG codes. 
Specifically, we use projective elliptic curves, projective hyper-elliptic curves 
and projective Hermitian curves in Weierstrass from. 



First, we consider the elliptic curve $\mathcal{E}_{a,b,c}$ with genus $g=1$ defined by the equation 
\begin{align}
    \mathcal{E}_{a,b,c}: y^2+ay=x^3+bx+c,
\end{align}
where $a,b,c\in \F_{2^h}$. 
Let $S_{a,b,c}$ be the set of $x$-component of the projective points of $\mathcal{E}_{a,b,c} $ over $\F_{2^h}$, 
i.e., 
\begin{align}\label{eq.S_{a,b,c}}
    S_{a,b,c}=\{\alpha \in \F_{2^h}\mid \exists~\beta \in \F_{2^h},~{\rm s.t.}~\beta ^2+a\beta =\alpha ^3+b\alpha +c\}.
\end{align}
For any $\alpha_i \in S_{a,b,c}$, it gives exactly two points with $x$-component $\alpha_i $ and 
we denote these two points corresponding to $\alpha_i $ by $P_i ^{(1)}$ and $P_i ^{(2)}$.
Then the set of all rational points of $\mathcal{E} _{a,b,c}$ over $\F_q$ is 
$\{P_i ^{(1)}\mid \alpha_i \in S_{a,b,c}\} \cup \{P_i ^{(2)}\mid \alpha_i \in S_{a,b,c}\} \cup \{P_\infty\}$.
Clearly, the numbers of rational points of the curve $\mathcal{E} _{a,b,c}$ is $N(\mathcal{E}_{a,b,c})=2|S_{a,b,c}|+1$, where $|S_{a,b,c}|$ denotes the cardinality of $S_{a,b,c}$.

\begin{theorem}\label{t5}
    Let $q=2^h$, $2e \mid h $ and $E=\{x^{2^e+1} \mid x\in \F_q^*\}$. 
    Put $U_1=\{\alpha \in E \mid Tr_{\F_q/\F_2}(\alpha ^3)=0\}$ 
    and $U_2=\{\alpha \in E \mid Tr_{\F_q/\F_2}(\alpha+\alpha^3)=0\}$, 
    where $Tr_{\F_q/\F_2}$ is the trace mapping from $\F_q$ to $\F_2$. 
    Then the following statements hold. 
    \begin{enumerate}
        \item [\rm (1)] For any $1 \leq n \leq |U_1|$ and $3 \leq k \leq \lfloor \frac{2n+2^e+1}{2^e+1}\rfloor $, 
        there exists an $[2n,k-1, \geq 2n-k+1]_q$ $e$-Galois SO AG code. 

        \item [\rm(2)] For any $1 \leq n \leq |U_2|$ and $3 \leq k \leq \lfloor \frac{2n+2^e+1}{2^e+1}\rfloor $, 
        there exists an $[2n,k-1, \geq 2n-k+1]_q$ $e$-Galois SO AG code. 
    \end{enumerate}
\end{theorem}
\begin{proof}
    (1) Consider the projective elliptic curve $\mathcal{E} _{1,0,0}:y^2+y=x^3.$ 
    According to Lemma \ref{l4}, $U_1$ is a subset of $S_{1,0,0}$ defined in Equation (\ref{eq.S_{a,b,c}}). 
    Put $U_0 =\{\alpha _1,\alpha _2,\ldots,\alpha _n\}\subseteq U_1$ and let $h(x)=\prod _{i=1}^n(x-\alpha_i)$ and $h'(x)=\sum_{j= 1}^{n} \prod _{i=1,i \neq j}^n(x-\alpha_i)$.
    We have $\gcd(p^{\frac{h}{2}}-1,\frac{q-1}{p^e+1})=p^e-1$ since $2e \mid h$.
    It follows from Lemma \ref{l5} that $\alpha_j \in \F_{p^e}$ for $1 \leq j \leq n$. 
    Then, by Lemma \ref{l2}, $h'(\alpha_j)= \prod _{i=1,i \neq j}^n(\alpha_j-\alpha_i) \in \F_{p^e} \subseteq E$, for $1 \leq j \leq n$.
    Hence, there exists $\beta_j \in \F_q^*$ such that $h'(\alpha _j)=\beta_j ^{p^e+1} $, for each $1 \leq j \leq n$.
    
    Let $P_i =P_i ^{(1)}+P_i^{(2)}$ for $1\leq i\leq n$. 
    Set $G=(k-1)P_\infty$, $D=(h(x))_0=P_1+P_2+\cdots+P_n$ and $w=\frac{dx}{h(x)}$. 
    Then $H=D-G+(w)=(2n)P_\infty-G$ and ${\Res}_{P_i^{(1)}}(w)={\Res}_{P_i^{(2)}}(w)=\frac{1}{h'(\alpha _i)} =\frac{1}{\beta _i^{2^e+1}}$, where $1 \leq i \leq n$.  
    Since $3 \leq k \leq \lfloor \frac{2n+2^e+1}{2^e+1}\rfloor $, we have $2^eG\leq H$. 
    Taking $\bm{v}=(\beta _1^{-1},\beta _1^{-1}, \beta _2^{-1},\beta _2^{-1},\ldots,\beta _n^{-1}, \beta _n^{-1})$, the desired result follows from Lemmas \ref{p.parameters} and \ref{l1}. 
    This completes the proof of the result (1).  

    (2) For the desired result (2), we consider the projective elliptic curve $\mathcal{E} _{1,1,0}: y^2+y=x^3+x$ 
    and the left proof is similar to the proof of (1) above. 
\end{proof}


Next, we consider the projective hyper-elliptic curve with genus $g=\frac{\sqrt{q}}{2}$ over $\F_q$ with $q=2^h$, 
which is given by the equation 
\begin{align}\label{eq.hyper-elliptic}
    \mathcal{H} :y^2+y=x^{\sqrt{q}+1}. 
\end{align}
From the discussion in \cite[Proposition 4.1]{F.qso}, the number of rational points of $\mathcal{H}$ is
\begin{align}\label{eq. numbers of hyper-elliptic rational points}
    N(\mathcal{H} )=2(\gcd((\sqrt{q}+1)(\sqrt{q}-1),q-1)+1)+1=2q+1.   
\end{align}
Noting that for any $\alpha_i \in \F_q$, there also exist two points 
$P_i^{(1)}$ and $P_i^{(2)}$ with $x$-component $\alpha_i$.
Then we have the following result. 
\begin{theorem}\label{t6}
    Let $q=2^h\geq 4$, where $h$ is even. 
    Let $N(\mathcal{H})$ be defined as in Equation (\ref{eq. numbers of hyper-elliptic rational points}). 
    Let $(n-1) \mid  (q-1)$ and $2n \leq N(\mathcal{H})-1$. 
    Then for any $1+\sqrt{q} \leq k \leq \lfloor \frac{2n+2^e+\sqrt{q}-1}{2^e+1}\rfloor $, 
    there exists an $[ 2n,k-\frac{\sqrt{q}}{2}, \geq 2n-k+1]_q$ $e$-Galois SO AG code.
      
\end{theorem}

\begin{proof}
    We consider the projective hyper-elliptic curve $\mathcal{H}$ defined in Equation (\ref{eq.hyper-elliptic}).
    Put $U=\{\alpha \in \F_q\mid \alpha ^n=\alpha \}=\{ \alpha _1, \alpha _2, \ldots,\alpha _n\}$
    and let $h(x)=\prod _{i=1}^{n}(x-\alpha _i )=x^n-x$. Then for any $\alpha _i \in U$, we have $h'(\alpha _i)=1$, where $1 \leq i \leq n$. 
    
    Let $P_i=P_i^{(1)}+P_i^{(2)}$ for $1\leq i\leq n$. 
    Set $G=(k-1)P_\infty$, $D=(h(x))_0=P_1+P_2+\cdots+P_n$ and $w=\frac{dx}{h(x)}$. 
    Then $H=D-G+(w)=(2n+\sqrt{q}-k-1)P_\infty$ and 
    ${\Res}_{P_i}^{(1)}(w)={\Res}_{P_i}^{(2)}(w)=1$.       
    Since $1+\sqrt{q} \leq k \leq \lfloor \frac{2q+2^e+\sqrt{q}-1}{2^e+1}\rfloor $, then $2^eG \leq H$. 
    By taking $\bm{v}=(1,1,\ldots,1)$, it follows from Lemmas \ref{p.parameters} and \ref{l1} that the desired result holds. 
    This completes the proof. 
\end{proof}


Finally, we employ projective Hermitian curves to construct Galois SO AG codes.  
Let $q=p^h$, where $h$ is even. Consider the projective Hermitian curve $\mathcal{T}$ 
with genus $g=\frac{q-\sqrt{q}}{2}$ given by 
\begin{align}
    \mathcal{T} :y^{\sqrt{q}}+y=x^{\sqrt{q}+1}.
\end{align}
Note that for any $\alpha_i \in \F_q$, there exist $\sqrt{q}$ rational points 
$P_i ^{(1)},P_i ^{(2)},\ldots, P_i ^{(\sqrt{q})}$ with $x$-component $\alpha_i$ \cite{F.qso}. 
Hence, the number of rational points of $\mathcal{T}$ is $N(\mathcal{T})=\sqrt{q}\cdot q+1=q^{\frac{3}{2}}+1$. 

\begin{theorem}\label{t7}
    Let $q=p^h$ and $2e \mid  h$, where $p$ is an odd prime and $0\leq e\leq h-1$. 
    If one of the following conditions holds: 
    \begin{enumerate}
        \item[\rm(1)] $n=tp^{aw}$, where $a \mid e$, $1 \leq t \leq p^a$ and $1 \leq w \leq \frac{h}{a}-1$; 
        \item[\rm(2)] $n=tp^{h-e}$, where $1 \leq t \leq p^e$;  
        \item[\rm(3)] $n=(t+1)\frac{q-1}{p^e+1}+1$, where $1 \leq t \leq p^e$ and $(p^{2e}-1) \mid \frac{q-1}{p^e+1}$;
        \item[\rm(4)] $n=\frac{t(q-1)}{p^e-1}+a$, where $a=\{0,1\}$ and $1 \leq t \leq p^e-1$;         
        \item[\rm(5)] $n=\frac{r(q-1)}{\gcd(q-1,x_2)}+a$, where $a=\{0,1\}$, $1 \leq r \leq  \frac{q-1}{\gcd(q-1,x_1)}$, $(q-1) \mid  {\lcm}(x_1,x_2)$ and ${\gcd}(x_2,q-1) \mid  x_1(p^e-1)$;         
        \item[\rm(6)] $n=rm+a$, where $a=\{0,1\}$, $m\mid (q-1)$, $1 \leq r \leq \frac{p^e-1}{m_1}$, $m_1=\frac{m}{\gcd(m,y)}$ and $y=\frac{q-1}{p^e-1}$,          
    \end{enumerate}
    then for any $1+q-\sqrt{q} \leq k \leq \lfloor \frac{\sqrt{q}n+p^e+q-\sqrt{q}-1}{p^e+1}\rfloor $, 
    there exists an $[\sqrt{q}n,k-\frac{q-\sqrt{q}}{2},\geq \sqrt{q}n-k+1]_q$ $e$-Galois SO AG code,
\end{theorem}

\begin{proof}
    Let $h(x)=\prod _{\alpha \in U}(x-\alpha)$, where $U$ is taken as one of the following cases: 
    \begin{enumerate}
        \item[\rm(1)] Put $U=\cup_{j=1}^{t}U_j$, where $U_j=K+\beta _i\eta $ is from \cite[Theorem III.10]{caomeng};
        \item[\rm(2)] Put $U=\cup_{i=1}^tU_i$, where $U_i=\{\alpha \in \F_q\mid Tr(x)=b_i\}$ for some $b_i \in \F_q$ is from \cite[Theorem 3.2]{liyang.gso}; 
        \item[\rm(3)] Put $U=U_N \cup \{\cup_{i=1}^t\alpha _iU_N\} \cup \{0\}$ 
        where $U_N=\{\alpha \in \F_q\mid \alpha ^N=1 \}$ is from Theorem \ref{t3}; 

       \item[\rm(4)] Put $U=\cup_{i=1}^tU_i$ if $a=0$, where $U_i=\{\alpha \in \F_q^*\mid {\rm Norm}(x)= x^{\frac{q-1}{p^e-1}}= b_i\}$ 
        for some $b_i \in \F_q^*$ is from \cite[Theorem 3.5]{liyang.gso} and 
        put $U=\cup_{i=1}^tU_i \cup \{0\}$ if $a=1$, where $U_i=\{\alpha \in \F_q^*\mid {\rm Norm}(x)= x^{\frac{q-1}{p^e-1}}= b_i\}$ for some $b_i \in \F_q^*$ 
        is from \cite[Theorem III.2]{caomeng}; 
        
        \item[\rm(5)] Put $U=\cup_{i=1}^rU_i$ if $a=0$, where $U_i=\{\xi_1^i\xi_2^j\mid j=1,2,\ldots, r_2\}$, 
        $\xi_1=w^{x_1}$ and $\xi_2=w^{x_2}$ are from \cite[Theorem 3.9]{liyang.gso} and 
        put $U=\cup_{i=1}^rU_i\cup\{0\}$ if $a=1$, where $U_i=\{\xi_1^i\xi_2^j\mid j=1,2,\ldots, r_2\}$, 
        $\xi_1=w^{x_1}$ and $\xi_2=w^{x_2}$ are from \cite[Theorem III.4]{caomeng}; 

        \item[\rm(6)] 
        
        Put $U=\cup_{i=1}^r\eta _iH$ if $a=0$, where $H=\langle\theta _1\rangle$, $V=\langle\theta _2\rangle$, 
        $\theta _1=w^{\frac{q-1}{m}}$, $\theta _2=w^{\frac{y}{m_2}}$ and $\eta _i$ is the left coset representative 
        of $V/H$ for $i=1,2,\ldots,\frac{p^e-1}{m_1}$ from \cite[Theorem 3.13]{liyang.gso} 
        and put $U=\cup_{i=1}^r\eta _iH$ if $a=1$, where 
        $H=\langle\theta _1\rangle$, $V=\langle\theta _2\rangle$, $\theta _1=w^{\frac{q-1}{m}}$, $\theta _2=w^{\frac{y}{m_2}}$ 
        and $\eta _i$ is the left coset representative of $V/H$ for $i=0,1,\ldots,\frac{p^e-1}{m_1}$ from \cite[Theorem III.8]{caomeng}.
    \end{enumerate}
    One can check that for any $\alpha _i \in U$, there exists $\beta _i\in \F_q^*$ such that $h'(\alpha _i)={\beta _i^{p^e+1}}$ for $1\leq i\leq |U|$.  
        
    Let $P_i=P_i ^{(1)}+P_i ^{(2)}+\cdots+P_i ^{(\sqrt{q})}$, where $1\leq i\leq n$. 
    Set $G=(k-1)P_\infty$, $D=(h(x))_0=P_1+P_2+\cdots+P_n$ and  $w=\frac{dx}{h(x)}$.  
    Then $H=D-G+(w)=(\sqrt{q}n+q-\sqrt{q}-k-1)P_\infty$ and ${\Res}_{P_i}^{(1)}(w)={\Res}_{P_i}^{(2)}(w)=\cdots={\Res}_{P_i}^{(\sqrt{q})}(w)=\frac{1}{h'(\alpha _i)} =\frac{1}{ \beta _i^{p^e+1}}$.   
    Since $1+q-\sqrt{q} \leq k \leq \lfloor \frac{\sqrt{q}n+p^e+q-\sqrt{q}-1}{p^e+1}\rfloor $, one has $p^eG\leq H$. 
    By taking $\bm{v}=(\underbrace{\beta _1^{-1},\ldots,\beta _1^{-1}}_{\sqrt{q}}, \underbrace{\beta _2^{-1},\ldots, \beta _2^{-1}}_{\sqrt{q}}, \ldots, \underbrace{\beta _n^{-1},\ldots,\beta _n^{-1}}_{\sqrt{q}})$, 
    the desired results then follows from Lemmas \ref{p.parameters} and \ref{l1}.
    This completes the proof. 
\end{proof}

\begin{remark}  
    We give some comparison results between Theorems \ref{t5}, \ref{t6}, \ref{t7} and some known conclusions.
    \begin{itemize}
        \item [\rm (1)] From special projective elliptic curves and projective hyper-elliptic curves over $\F_{2^h}$ with even $h$, 
        we derive two new classes of Galois SO AG codes in Theorems \ref{t5} and \ref{t6}.  
        On one hand, considering $e=0$, 
        Theorems \ref{t5} and \ref{t6} respectively yield two classes of Euclidean self-dual AG codes 
        with  parameters $[2n,n,\geq n]_{2^h}$ and $[2n,n,\geq n-\frac{\sqrt{q}}{2}+1]_{2^h}$, 
        which are the same with \cite{sok.2021 sd}.
        Hence, Theorems \ref{t5} and \ref{t6} indeed generalize the results in \cite{sok.2021 sd}. 
        On the other hand, comparing with Table \ref{tab:1}, we find Class 17 has similar conditions with Theorem \ref{t6}.
        However, the length of the codes in our construction are twice than Class 17 over $\F_{2^h}$, where $h$ is even.
        Hence, the Galois SO codes with lengths larger than $q$ from Theorem \ref{t6} are new.

        \item [\rm (2)] Theorem \ref{t7} contains nine explicit constructions of Galois SO AG codes and 
        the lengths of these codes are $\sqrt{q}$ times of those of Galois SO codes constructed in \cite{caomeng,liyang.gso} and Theorem \ref{t3}.
        Note also that for $e$-Galois SO AG codes constructed in Theorem \ref{t7}\ (3), the value of $e$ much be less than $\frac{h}{2}$ since $(p^{2e}-1) \mid \frac{q-1}{p^e+1}$.
        As a result, if we take $e=\frac{h}{2}$ with even $h$, we can immediately get eight classes of Hermitian SO AG codes with lengths larger than $q$.  
        As far as we know, the derived Hermitian SO AG codes with lengths larger than $q$ are new. 
    \end{itemize}
\end{remark}


\begin{example}\label{e}
    Take $p=2$ and $h=8$, then all cases satisfying $2e \mid h$ are $e=1$, $e=2$ and $e=4$.
    Note that in this case, $4$-Galois SO codes are just Hermitian self-orthogonal codes.
    For the general $e$-Galois SO codes, we list some $1$-Galois SO AG codes and $2$-Galois SO AG codes as follows.
    \begin{enumerate}
        \item[\rm(1)] For $e=1$, we obtain 
        $[36,9,\geq 20]_{2^8}$, 
        $[104,k-8,\geq 105-k]_{2^8}\ (17 \leq k \leq 40)$, 
        $[172,k-8, \geq 173-k ]_{2^8}\ (17 \leq k \leq 63)$, $[512,k-8, \geq 513-k]_{2^8}\ (17 \leq k \leq 176)$ 
        $1$-Galois SO AG codes from Theorem \ref{t6}.

        \item[\rm(2)] For $e=2$, we obtain 
        $[104,k-8,\geq 105-k]_{2^8}\ (17 \leq k \leq 24)$, 
        $[172,k-8, \geq 173-k ]_{2^8}\ (17 \leq k \leq 38)$, $[512,k-8, \geq 513-k]_{2^8}\ (17 \leq k \leq 106)$
        $2$-Galois SO AG codes from Theorem \ref{t6}.
        
    \end{enumerate}
\end{example}

\begin{example}
    Take $p=3$ and $h=8$. Similar to Example \ref{e}, we only list some $1$-Galois SO AG codes and $2$-Galois SO AG codes as follows.
    \begin{enumerate}
        \item[\rm(1)] For $e=1$, we obtain $[177147,k-3240,\geq 17148-k]_{3^8}\ (6841 \leq k \leq 45907)$,
                      $[354294,k-3240,\geq 354295-k]_{3^8}\ (6841 \leq k \leq 90194)$,
                      $[531441,k-3240, \geq 531442-k]_{3^8}\ (6841 \leq k \leq 134480)$
                      $1$-Galois SO AG codes from Theorem \ref{t7}.

        \item[\rm(2)] For $e=2$, we obtain $[59049,k-3240,59050-k]_{3^8}\ (6841 \leq k \leq 6553)$,
                      $[118098,k-3240,\geq 118099-k]_{3^8}\ (6841 \leq k\leq 12458)$,
                      $[177147,k-3240,\geq 177148-k]_{3^8}\ (6841 \leq k\leq 18363)$,
                      $[236196,k-3240,\geq 236197-k]_{3^8}\ (6841 \leq k\leq 24268)$,
                      $[295245,k-3240,\geq 295246-k]_{3^8}\ (6841 \leq k\leq 30173)$,
                      $[354294,k-3240,\geq 354295-k]_{3^8}\ (6841 \leq k\leq 36078)$,
                      $[413343,k-3240,\geq 413344-k]_{3^8}\ (6841 \leq k\leq 41983)$,
                      $[472392,k-3240,\geq 472393-k]_{3^8}\ (6841 \leq k\leq 47888)$,
                      $[531441,k-3240,\geq 531442-k]_{3^8}\ (6841 \leq k\leq 53792)$
                      $2$-Galois SO AG codes from Theorem \ref{t7}.
    \end{enumerate}

\end{example}

\section{Conclusions}\label{sec5}
  In this paper, we focus on the gap on the existence and constructions of general Galois SO AG codes. 
  We give a criterion that determines when an AG code becomes Galois SO. 
  Based on this criterion, we get many new (MDS) Galois SO AG codes from projective lines 
  and some special projective curves. An embedding method that allows us to derive new MDS Galois 
  SO codes from known MDS Galois SO AG codes constructed from projective lines is also presented. 
  By some specific comparisons, one can see that our constructions are new with respect to general Galois inner products 
  and generalize some results in \cite{sok.2021 sd,sok.2021 hso} with respect to special Euclidean and Hermitian inner products.

  Note that the so-called Galois hulls of generalized Reed-Solomon codes and Goppa codes 
  have been deeply studied in \cite{caomeng,FXL,Goppa.2023,liyang.gso,liyang.MDS,Goppa.2022}.  
  Recall that generalized Reed-Solomon codes and Goppa codes are two special subclasses of AG codes. 
  Hence, it would be interesting to study the Galois hulls of general AG codes and their performance in 
  constructing entanglement-assisted quantum error-correcting codes in the future.

  \section*{Declarations}

  \noindent\textbf{Data availability} No data are generated or analyzed during this study.  \\
  
  \noindent\textbf{Conflict of Interest} The authors declare that there is no possible conflict of interest. 
 
  \section*{Acknowledgments}
  This research is supported by the National Natural Science Foundation of China under Grant No. 12171134 and U21A20428.

\end{sloppypar}
\end{document}